\documentclass[10pt,journal, twocolumn]{IEEEtran}
\usepackage[utf8]{inputenc}
\usepackage{amssymb,amsfonts,amsthm,bm}
\usepackage{url}
\usepackage{color}
\usepackage{graphicx}
\usepackage{multirow}
\usepackage{cite}
\usepackage{mathtools}
\usepackage{bigstrut} %to put more space for cells inside a table
\usepackage{tikz,tikz-3dplot}
\usetikzlibrary{3d}
\usetikzlibrary{decorations.pathreplacing,}
\tikzset{radiation/.style={{decorate,decoration={expanding waves,angle=90,segment length=4pt}}},
         antenna/.pic={
        code={\tikzset{scale=5/10}
            \draw[semithick] (0,0) -- (1,4);% left line
            \draw[semithick] (3,0) -- (2,4);% right line
            \draw[semithick] (0,0) arc (180:0:1.5 and -0.5);
            \node[inner sep=4pt] (circ) at (1.5,5.5) {};
            \draw[semithick] (1.5,5.5) circle(8pt);
            \draw[semithick] (1.5,5.5cm-8pt) -- (1.5,4);
            \draw[semithick] (1.5,4) ellipse (0.5 and 0.166);
            \draw[semithick,radiation,decoration={angle=45}] (1.5cm+8pt,5.5) -- +(0:2);
            \draw[semithick,radiation,decoration={angle=45}] (1.5cm-8pt,5.5) -- +(180:2);
  }}
}
\tikzset{
         airplane/.pic={
        code={\tikzset{scale=5/10}
            \draw[semithick] (0,0) -- (1,4);% left line
            \draw[semithick] (3,0) -- (2,4);% right line
            \draw[semithick] (0,0) arc (180:0:1.5 and -0.5);
            \node[inner sep=4pt] (circ) at (1.5,5.5) {};
            \draw[semithick] (1.5,5.5) circle(8pt);
            \draw[semithick] (1.5,5.5cm-8pt) -- (1.5,4);
            \draw[semithick] (1.5,4) ellipse (0.5 and 0.166);
  }}
}
\newtheorem{thm}{Theorem}
\def\x{{\mathbf x}}

\def\h{{\mathbf h}}

\def\D{{\mathbf D}}

\def\e{{\mathbf e}}

\def\x{{\mathbf x}}
\def\0{{\mathbf 0}}

\def\U{{\mathbf U}}

\def\Eb{{\mathbf E}}
\def\i1{{\mathbf 1}}

\def\CN{{\mathcal{CN}}}

\begin{document}
	\title{Airplane-Aided Integrated Next-Generation Networking}

\author{ Muralikrishnan Srinivasan$^{1}$, Sarath Gopi$^{2}$, Sheetal Kalyani$^{2}$, Xiaojing Huang$^{3}$, Lajos Hanzo$^{4}$ {

\thanks{1. Muralikrishnan Srinivasan is with ETIS UMR8051, CY University, ENSEA, CNRS, Cergy, France.
(Email:muralikrishnan.srinivasan@ensea.fr)}
\thanks{2. Sarath Gopi and Sheetal Kalyani are with the Dept. of Electrical Engineering, Indian Institute of Technology, Madras, India.
(Emails:{ee14d007@ee, skalyani@ee}.iitm.ac.in).}
\thanks{3. Xiaojing Huang is with the University of Technology Sydney. (Email: Xiaojing.Huang@uts.edu.au)}
\thanks{4. Lajos Hanzo is with the School of Electronics and Computer
Science, University of Southampton.(Email: hanzo@soton.ac.uk)}
\thanks{ Muralikrishnan Srinivasan and Sarath Gopi are co-first authors.}}
\thanks{L. Hanzo would like to acknowledge the financial support of the Engineering and Physical Sciences Research Council projects EP/P034284/1 and EP/P003990/1 (COALESCE) as well as of the European Research Council's Advanced Fellow Grant QuantCom (Grant No. 789028).}
\thanks{This work was also supported in part by the Australian Research Council Discovery Project (DP200101532).}
\thanks{\textbf{This article has been accepted for publication in IEEE
Transactions on Vehicular Technology, but has not been fully edited. Content may change prior to final publication. Citation information: DOI 10.1109/TVT.2021.3098098, IEEE
Transactions on Vehicular Technology}}
} 

	\maketitle
	\begin{abstract}
     A high-rate yet low-cost air-to-ground (A2G) communication backbone is conceived for integrating the space and terrestrial network by harnessing the opportunistic assistance of the passenger planes or high altitude platforms (HAPs) as mobile base stations (BSs) and millimetre wave communication. The airliners act as the network-provider for the terrestrial users while relying on satellite backhaul.  
     Three different beamforming techniques relying on a large-scale planar array are used for transmission by the airliner/HAP for achieving a high directional gain, hence minimizing the interference among the users. Furthermore, approximate spectral efficiency (SE) and area spectral efficiency (ASE) expressions are derived and quantified for diverse system parameters.
\end{abstract}

	\section{Introduction}
\par Next-generation wireless standards are expected to cope with increased traffic demands and support emerging applications even in remote locations such as rural hinterlands, mountains, deserts and even for vessels such as cruise-ships in the oceans \cite{Huang2019}. However, the operational fifth-generation (5G) standards have predominantly been designed for terrestrial communications. One of the promising techniques of augmenting cellular communication is through air-based platforms such as unmanned aerial vehicles (UAV) or high-altitude platforms (HAP). Hence, extensive research has been dedicated to the design, channel modelling, and security of UAV-based cellular communications \cite{Khawaja2019, Fotouhi2019, Zeng2019, Kurt2020}, as well as to their performance analysis \cite{Mozaffari2016, Azari2017, Chetlur2017, Azari2017a, Galkin2017, Ono2016, Lyu2016, Enayati2019, Arum2019, Xu2019}.   Sakhaee and Jamalipour~\cite{sakhaee2006global} along with Kato~\cite{sakhaee2006aeronautical} showed as early as 2006 that the probability of finding at least two but potentially up to dozens of aircraft capable of establishing an AANET above-the-cloud is close to $100\%$. It was inferred by investigating a snapshot of flight data over the United States (US). They also proposed a quality of service (QoS) based so-called multipath Doppler routing protocol by jointly considering both the QoS and the relative velocity of nodes in order to find stable routing paths.

\par Integrating the aerial networks with the terrestrial networks has the potential of increasing both the data rate and the coverage quality of terrestrial networks \cite{Reynaud2011, Zhang2017, Dinc2017, Azari2017b, Qiu2019, Cheng2020}.  There have also been some attempts to integrate the space networks with terrestrial networks or to provide Internet coverage for airliner \cite{Evans2005, Cianca2005, Wang2017, Lagunas2015, Kandeepan2011, Zhang2018, Liu2018, Xu2019a}. The applicability of these prior contributions are tabulated in \ref{tab:works}. However, most of these contributions rely on reusing the existing long-term evolution (LTE) bands, which are already congested in the sub-6GHz bands \cite{Simonite2018, 3GPP}.   Therefore, the creation of a high-capacity integrated space terrestrial network (ISTN) or a  space-air-ground integrated network (SAGIN) is still elusive at the time of writing both due to the bandwidth limitation of aerial backbones and owing to the limited area spectral efficiency (ASE) of the air-to-ground (A2G) systems, given their large footprint on the ground.   Therefore, it is imperative to explore new architectures integrating the existing terrestrial networks with space networks. 
 
	\begin{table}[!h]
	    \centering
	    \begin{tabular}{||c|r||}
	    \hline
	    \hline
	        Works & Applicability\\
	        \hline
	        \hline
	        \cite{Reynaud2011} & Emergency networks 	\\
	        \hline
	        \cite{Zhang2017} & Vehicular networks\\
	         \hline
	         \cite{Dinc2017} & Broadbad connectivity\\
	         \hline
	         \cite{Azari2017b} & Cellular networks\\
	         \hline
	         \cite{Qiu2019} & 5G Cellular networks\\
	         \hline
	         \cite{Cheng2020} & Secure networks \\
	         \hline
	         \cite{Evans2005} & Multi-media systems  \\
	         \hline
	         \cite{Cianca2005} & Integration of space and HAP systems  \\
	         \hline
	         \cite{Wang2017} & Emergency networks  \\
	         \hline
	         \cite{Lagunas2015} & Cognitive communications \\
	         \hline
	         \hline
	         \hline
	    \end{tabular}
	    \caption{Integrated aerial-terrestrial networks in Sub-6GHz band}
	    \label{tab:works}
	\end{table}

\subsection{ Related Contributions:} 
  
A critical cornerstone of the next-generation systems is the potential exploitation of millimetre wave (mmWave) carrier frequencies to benefit from their broad unused spectrum. Hence the authors of \cite{Khawaja2017, Gapeyenko2018, Zhao2018, Cuvelier2018, Dutta2019, Popoola2020} have investigated the challenges of mmWave based A2G and air-to-air (A2A) communications. For example, the two-ray propagation model's applicability in different scenarios employing UAVs was explored in \cite{Khawaja2017}. The pros and cons of UAV-BS in complementing the mmWave backhaul were demonstrated in \cite{Gapeyenko2018}. The concept of mmWave A2A networks was first explored by Cuvelier and Heath \cite{Cuvelier2018}, while mmWave based HAPs in terrestrial transmissions were studied in \cite{Dutta2019, Popoola2020}.

\par As a further development, Huang {\em et al.}~\cite{Huang2019} have proposed the ISTN concept relying on civil airliner networks and mmWave communication to form a high-capacity yet low-cost A2A and A2G communications backbone employing high-gain antenna arrays. In this concept, the airliners act as an efficient network-provider for terrestrial users, since the distance from the planes to the ground is much shorter than that from the satellite. Furthermore, to provide A2G cellular coverage for small cells that can support high ASE, adaptive beamforming is proposed. In areas where the civil-airliners cannot be used, dedicated HAPs would be used as the backbone. A similar topology is presented in \cite[Fig 1]{Zhang2019}. 

\begin{table*}[t]
 \begin{center}
			\begin{tabular}{|l|l|l|l|l|l|l|l|l|l|l|}
				\hline
				& Our Scheme & \cite{El-Jabu2001} & \cite{Khan2010}  &  \cite{Huo2018} &\cite{Zhong2019} &\cite{Zhong2020} & \cite{Interdonato2020} & \cite{Zhu2019c}  & \cite{Vaezy2020} & \cite{Xu2017three}\\
				\hline
				mmWave &$\checkmark$ & & &$\checkmark$&$\checkmark$ &$\checkmark$ & 	&$\checkmark$ &$\checkmark$ &$\checkmark$    \\
				\hline
				Airliner/UAV backbone &  Airliner & UAV & & UAV& UAV & UAV & & UAV & UAV & HAP\\
								\hline 
				Adaptive null steering & $\checkmark$ & & $\checkmark$&  & && $\checkmark$ & & & \\
				\hline 
				Channel unaware & & & & & & & & & &
				\\ 
				precoding &$\checkmark$ & & $\checkmark$ & &&$\checkmark$ & $\checkmark$ &  & &\\
				\hline
				Rician channel &$\checkmark$ & & & &&$\checkmark$ & & & &$\checkmark$ \\
				\hline
				
				Simulated Metric & ASE/SE & SE & Array  & SE & SE & SE& SE & Beam  & Beam & SE\\
				 &  & & response  &  & &  &  & coverage  & coverage & \\
				\hline
				Theoretical expressions & ASE/SE &  & & & & SE &  & & & SE\\
				\hline
			\end{tabular}
		\caption{Comparison of the proposed scheme with existing works}
		\end{center}
		\label{tab:CompTable}
	\end{table*}

\subsection{ Design challenges:}
To actually design such a high-capacity airline-aided integrated network, several challenges have to be addressed. For example, a cruising airliner maintains an altitude of at least $10~km$ from the ground, while solar-charged unmanned aircraft are envisioned to circle above $20~km$ for avoiding civilian planes. The mmWave channel suffers from substantial pathloss owing to raindrops, high-absorption and other atmospheric effects, especially at a carrier frequency of $73.5$ GHz. Another challenge to overcome is the huge channel estimation overhead, which results from the rapidly fluctuating high-Doppler channel between the cruising airliner and ground users. Hence, a careful selection of the channel model, antenna dimensions, Rician factor and other system parameters is required for investigating a realistic stand-alone model.
 
 \subsection{Contributions}
To overcome the above-mentioned challenges, we design a high-performance system having a high data-rate and ASE and provide theoretical performance guarantees with the aid of approximate expressions. We consider a planar-array aided stand-alone airliner/HAP in a macro-cell communicating with the terrestrial BS/users.  Although a strong line-of-sight (LoS) component exists between the airplane and the ground users, the non-line-of-sight (NLoS) component fluctuates drastically over time.   Hence, to achieve a high directional gain while minimizing the users' interference, we propose three different channel-agnostic transmit precoding schemes for avoiding the massive pilot overhead required for estimating the channel-state information (CSI) or the NLoS component. Explicitly, our Transmit Precoders (TPC) rely only on the users' position relative to the airliner and they are quite robust to incorrect Doppler compensation and position vector mismatches.

\par We also derive approximate expressions for the SE/ASE of the users. Furthermore, the proposed schemes are evaluated through extensive simulations, and its performance is compared to the analytically obtained values. Additionally, depending on the dimensions of the planar array and of the LoS factor, the ASE achieved by our system becomes several times higher than that of conventional terrestrial networks \cite{Alouini1999, Xin2015, Li2016, Ding2015} capable of providing data rates on the order of several Gbps.    The authors of \cite{El-Jabu2001, Khan2010, Huo2018, Zhong2019, Zhong2020, Interdonato2020, Zhu2019c, Vaezy2020, Xu2017three} have considered 3D beamforming in the context of UAV communications, where most of them tended to rely either on UAVs flying at a modest altitude or on channel-aware TPCs. Specifically, the authors of \cite{Xu2017three} designed a TPC relying on the effective channel matrix between the access point (AP) and the users. Therefore, a dedicated downlink channel estimation phase is required by the TPC. By contrast, in our case, the channel between the airliner and the user is dominated by the LoS component due to the airliner's altitude. Typically the aeronautical model for such a scenario has a strong LoS path and a much weaker NLoS ground-reflected path \cite{haas2002aeronautical}. Hence using a TPC vector relying on the estimated channel matrix is counterproductive. Furthermore, in such a highly mobile scenario, accurate channel estimation requires a potentially excessive pilot overhead. Explicitly, compared to the pedestrian walking across a mmWave cell at $5~ km/hr$, a plane travelling at $1000~km/hr$ would require a $200$ times higher pilot overhead. Even if we disregard the huge training overhead requirement, the CSI is prone to estimation error.

Against the above backdrop, we boldly contrast our novel contributions to the prior art in Table II. The theoretical analysis of the proposed system and our extensive simulations indicate that our design leads to high capacity airplane-aided integrated networks that are eminently suitable for filling the coverage-holes of next-generation wireless systems.

The rest of the paper is structured as follows. In Section II, the proposed system design is discussed in detail, along with different beamforming. In Section III, theoretical expressions are derived for the ASE/SE using the popular use and forget bound. In Section IV, our simulation results and interesting design guidelines are discussed, while in Section V, some future research directions are provided.

\section{Proposed System Design}
	
	\begin{figure*}
	\centering
	\resizebox{0.70\textwidth}{!}
	{
		\begin{tikzpicture}
		
		 %\tikzset{xzplane/.style={canvas is xz plane at y=#1,very thin}}
		 %\tikzset{xyplane/.style={canvas is xy plane at z=#1,very thin}}

		\begin{scope}[canvas is xz plane at y=-1]

		\node (x)[] at (6, 6) {};
 		\node (U1)[label=below:{\Large{User}}] at (6.5, 5.5) {};
 		\path (6.3,5.8) pic[scale=0.30] {antenna};
		\node (I1)[] at (6, 10) {};
		\path (6.3,10.2) pic[scale=0.30, color=blue] {antenna};
		\node (I2)[] at (9.5, 8) {};
		\path (9.5,8.1) pic[scale=0.30, color=blue] {antenna};
		\node (I3)[] at (9.5, 4) {};
		\path (9.6,4.2) pic[scale=0.30, color=blue] {antenna};
		\node (I4)[] at (6, 2) {};
		\path (6,2.3) pic[scale=0.30, color=blue] {antenna};
		\node (I5)[] at (2.5, 4) {};
		\path (2.3,4.5) pic[scale=0.30, color=blue] {antenna};
		\node (I6)[] at (2.5, 8) {};
		\path (2,8) pic[scale=0.30, color=blue] {antenna};

		\draw[dashed, very thick] (x) circle (3cm);
		\draw[dashed, very thick] (x) circle (5cm);
		
		\filldraw[fill=lightgray, fill opacity =0.3, draw=black] (x) circle (1cm);
		\filldraw[fill=lightgray, fill opacity =0.3, draw=black] (I1) circle (1cm);
		\filldraw[fill=lightgray, fill opacity =0.3, draw=black] (I2) circle (1cm);
		\filldraw[fill=lightgray, fill opacity =0.3, draw=black] (I3) circle (1cm);
		\filldraw[fill=lightgray, fill opacity =0.3, draw=black] (I4) circle (1cm);
		\filldraw[fill=lightgray, fill opacity =0.3, draw=black] (I5) circle (1cm);
		\filldraw[fill=lightgray, fill opacity =0.3, draw=black] (I6) circle (1cm);

		%\filldraw (x) circle (2pt) ;
 		%\filldraw (U1) circle (2pt) ;
		%\filldraw (I1) circle (1pt) ;
		%\filldraw (I2) circle (2pt) ;
		%\filldraw (I3) circle (2pt) ;
		%\filldraw (I4) circle (2pt) ;
		%\filldraw (I5) circle (2pt) ;
		%\filldraw (I6) circle (2pt) ;
		
		%\draw[very thick] (x.center) -- (U1.center) node[pos=0.25, below]{$r_1$};
		%\draw[->, very thick] (I4.center) -- (11,6) node[pos=0.95, above]{\footnotesize{micro-cell radius $r$}};
		\draw[-, thick] (x.center) -- (I6.center) node[pos=0.25, left]{\small{Reuse Distance $D$}};
		
		\node (MCC)[label=left:{\Large{Macro-cell center}}] at (0, 0) {};
		\filldraw (MCC) circle (2pt) ;
        \draw [red,thick,domain=-35:100] plot ({0+15*cos(\x)}, {0+15*sin(\x)});
        \draw [very thick,domain=0:45] plot ({0+1*cos(\x)}, {0+1*sin(\x)});

       \draw[->, dotted] (MCC.center) -- (0, 15);
       \node (Y)[label=right:{\Large{y}}] at (0, 15) {};
        \draw[->, dotted] (MCC.center) -- (15, 0) node[pos=0.75, above]{\Large{macro-cell radius $R$}};
       	\node (X)[label=right:{\Large{x}}] at (15, 0) {};
        \draw[->, dotted] (MCC.center) -- (U1.center);
		\end{scope}

		\begin{scope}[canvas is xz plane at y=3]
		\node (A)[label=left:{\Large{Origin}}] at (0, 0) {\includegraphics[scale=0.20]{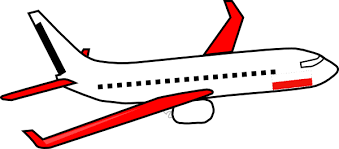}};
        \filldraw (A) circle (2pt) ;
		\draw[-, dotted] (U1.center) -- (A.center);
		\draw[-, dashed, very thick] (MCC.center) -- (A.center);
		\end{scope}
		
		\begin{scope}[canvas is xy plane at z=0]
		\draw [very thick,domain=-90:-55] plot ({0+1*cos(\x)}, {3+1*sin(\x)});
		\node (zen)[label=right:{\Large{$\theta^z$ zenith angle}}] at (0, 1.7) {};
		\node (azi)[label=right:{\Large{$\theta^a$ azimuth angle}}] at (0, -0.4) {};
		\draw[->, dotted] (A.center) -- (0, 5) node[pos=0.95, left]{\Large{z}};
		\end{scope}
		
		\end{tikzpicture}
		}
		\caption{The Airliner/HAP is assumed to be at the origin. The macro-cell centre is the point on the ground directly below the airliner. A circular macro-cell of radius $R$ is considered around the macro-cell centre. The zenith-azimuth angle pair $(\theta^z, \theta^a)$ of the 'User' in a micro-cell of interest (MCI), is marked with respect to the airliner. Each micro-cell has a radius $r$. The first tier of $N_I$ interfering micro-cells is shown at a reuse distance $D$. In this figure, $N_I=6$.} 
		\label{fig1}
	\end{figure*}
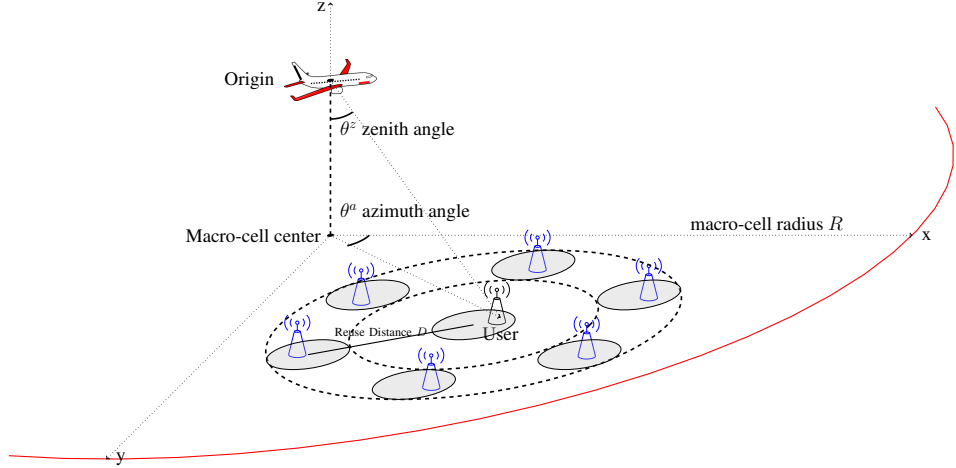

	In this section, we design the ISTN backbone and propose a system model using various adaptive beamforming for supporting high data-rates and seamless connectivity. Consider a circular macro-cell of radius $R$ with an airliner at its centre at an altitude $H_t$. The radius $ R $ is chosen to be at least $5~km$, so a minimum inter-airliner distance of $10~km$ is maintained. In remote locations outside the regular flight path, HAPs can be installed for providing seamless connectivity to satellites. Free-space optical (FSO) links connect them to Low-Earth orbit (LEO)/ medium-Earth orbit (MEO) satellites as their high-speed backhaul.
	
	Each airliner/HAP is equipped with a planar antenna having $M \times M $ equally spaced elements. Without loss of generality, it can be assumed that the antenna elements are parallel to the ground and the centre of the antenna is the origin $(0,0,0)$. The macro-cell is further divided into several tightly packed micro-cells of radius $r << R$. Each micro-cell supports a single time-frequency block. The users can either be a cellular user equipment (CUE) or even an LTE base station, which in turn supports several UEs. The intended user is at position $(x_0,y_0,-H_t)$, and the micro-cell containing it is referred to as the micro-cell of interest (MCI).
	
	Assume that there are $N_1$, $N_2$ ...$N_J$ interfering micro-cells in the first $J$ tiers using the same time-frequency block. The centres of these micro-cells are located at distances $D$, $2D$, ..., $JD$ from the MCI, where $D$ denotes the reuse distance. Let $N_I=N_1+N_2+...+N_J$ be the total number of interfering cells and let the coordinates of these interfering users be  $(x_i,y_i,-H_t)$ $i=1,..., N_I$. All the users are assumed to have a single antenna. Let $(\theta_i^z, \theta_i^a)$ represent the zenith and azimuth angle pair for the $i^{th}$ user, which are:
	\begin{align}
	\theta_i^z =\tan^{-1}\left(\frac{\sqrt{x_i^2+y_i^2}}{-H_t}\right), \quad i=0,..., N_I
	\label{thetaiz}
	\end{align}
	 and
	\begin{align}
	\theta_i^a= \tan^{-1}\left(\frac{y_i}{x_i}\right), \quad i=0,..., N_I.
	\label{thetaia}
	\end{align}		 
	Here $i=0$ represents the user under consideration, while $i=1,...,N_I$ represent the interferers.
	The entire system is shown in Fig. \ref{fig1}.

     Note that the typical cruising airliner altitudes are in the $9-12~km$ range. At such distances, the mmWave channels' attenuation, say at $73.5 ~ GHz $, is significant. Existing contributions, such as \cite{Huo2018, Zhong2019, Zhong2020, Zhu2019c, Vaezy2020}, which deal with aerial mmWave networks, consider only low-flying UAVs or low altitude platforms (LAPs) and hence suffer from relatively low attenuation. In the absence of beamforming at the transmitter, the users in the macro-cell suffer from mutual interference, resulting in a reduced data-rate. Therefore, to reduce the mutual interference amongst the users and improve the spectral efficiency, some form of adaptive beamforming must be used by the airliner's planar array. Furthermore, the adaptive beamforming schemes must be robust to incorrect Doppler compensation and position vector mismatches. We therefore propose three different beamforming or TPC schemes at the Airliner/HAP. The design of the beamforming vectors is described in the subsequent paragraphs.

    \subsection{Null-steered beamforming (NSB) design}
    \label{nullbf}
NSB relies on signal processing techniques for creating transmit nulls and maxima in the undesired and desired receivers' directions, respectively, for mitigating the interference \cite{Khan2010, Litva1996, Yoo2003}. 
 For $i=0,1,.., N_I$, let $\e_{i}$ be the $M^2 \times 1$ vector representation of the $i$th user's steering vector with respect to each of the components of the planar array. The three-dimensional Cartesian co-ordinates of the $(m,n)$th component of the planar array are represented by $( x_m,y_n, 0)$, where $x_m = \left[-\frac{M-1}{2} + (m-1)\right]\frac{\lambda}{2} $ and $y_n = \left[-\frac{M-1}{2} + (n-1)\right]\frac{\lambda}{2}$, for $m=1,..M$ and $n=1,...,M$. Note that the inter-elemental spacing is $\lambda/2$, where $\lambda$ is the wavelength of the carrier. Thus the entry of $\e_{i}$, which corresponds to the position of the user with respect to the $(m,n)$th element of the planar array, is given by $\exp\left[j\frac{2\pi}{\lambda}\left(x_m \psi_i^x + y_n\psi_i^y\right)\right]$, where we have
	\begin{align}
	\psi_i^x = \sin \theta_i^z \cos \theta_i^a,
	\label{psix}
	\end{align}
	and 
	\begin{align}
	\psi_i^y =  \sin \theta_i^z \sin \theta_i^a,
	\label{psiy}
	\end{align}
 while $\theta_i^z$ and  $\theta_i^a$ are defined in (\ref{thetaiz}) and (\ref{thetaia}), respectively, which are functions of the user location. Now the  null-steered beamforming vector $\tilde{\e}_i$ used by the airliner is 
 \begin{align}
	\tilde {\e}_i = \e_{i} - \Eb_i \left(\Eb_i^H\Eb_i\right)^{-1}\Eb_i^H\e_{i}, \forall i=0,1,..., N_I,
	\label{p5eqn2}
	\end{align}
	where $\Eb_i$ is the matrix whose columns are the steering vectors, except for $\e_i$, which is given by:
	\begin{align}
	\Eb_i = \left[\e_0~\e_2~...~\e_{i-1}~\e_{i+1}~...~\e_{N_I}\right].
	\label{p5eqn3}
	\end{align}
	This can be obtained by solving the following  optimization problem \cite{van2004optimum}:
	\begin{equation}\label{opt1}
	    \begin{aligned}
	    	&\underset{\tilde{\e}_i}{\min} \|\tilde{\e}_i - \e_{i}\|^2\\ ~\text{such that}~ & \tilde{\e}_i^H\Eb_i = \bm{0}.
	    \end{aligned}
	\end{equation}
NSB may also be interpreted as a zero-forcing precoder, where the precoding vectors only require the knowledge of the user location.
 
\subsection{Null-Steered Beamformer with Derivative Constraints (NSB-D)}
  The NSB detailed in Section \ref{nullbf} can be extended by adding additional derivative constraints \cite[3.7.2]{van2004optimum}. The higher order derivatives of the directional pattern in the directions of nulls are set to zero to broaden the beam-widths, to make the TPC robust to mismatches in the steering vector. In this paper, we explicitly add the following first-order derivative constraints to (\ref{opt1}):
\begin{align}
\bar{\e}_i^a =\frac{\partial}{\partial \theta_i^a} \e_i &= 0~\forall~i=0, 1, ..., N_I \nonumber \\
\bar{\e}_i^z =\frac{\partial}{\partial \theta_i^z} \e_i &= 0~\forall~i=0, 1, ..., N_I.
\label{p5eqnn4}
\end{align}
Solving the optimization problem yields the beamforming vector $\tilde{\e}_i$ as \cite{van2004optimum}:
\begin{align}
 \tilde{\e}_i = \e_i - \tilde \Eb_i\left(\tilde \Eb_i^H\tilde \Eb_i\right)^{-1}\tilde \Eb_i^H\e_i~\forall~i=0, 1, ..., N_I,
\label{p5eqnn6}
\end{align}
where
\begin{align}
\tilde \Eb_i = \left[\e_0 ~...~ \e_{i-1}~ \e_{i+1}~ ...~ \e_{N_I}~\bar{\e}_0^a ~\bar{\e}_0^z~...~\bar{\e}_{N_I}^a~\bar{\e}_{N_I}^z\right].
\label{p5eqnn5}
\end{align}
Note that (\ref{p5eqnn6}) is similar to (\ref{p5eqn3}) except that $\Eb$ is replaced with $\tilde \Eb$, which includes the derivative constraints of all the users. 
 
\subsection{Minimum-Power Distortionless Response Beamformer (MPDRB)}
  Another popular beamformer is the MPDRB, where the total output power is minimized subject to a distortionless constraint \cite{van2004optimum}. In other words, one can determine the steering vectors $\tilde{\e}_i$ so that the total power $\|\tilde{\e}_i^H \hat \Eb_i\|^2$ subject to a distortionless constraint is minimized, which is formulated as :
\begin{equation}
    \begin{aligned}
    \underset{\tilde{\e}_i}{\min} ~  \|\tilde{\e}_i^H \hat \Eb_i\|^2 ~\text{such that}~  \tilde{\e}_i^H\e_{i}= \bm{1},
    \end{aligned}
\end{equation}
where $\hat \Eb_i= \left[\e_0~\e_1~...~\e_{i-1}~\e_i ~\e_{i+1}~...~\e_{N_I}\right]$. Solving the optimization problem results in \cite{van2004optimum}:
\begin{equation}
 \tilde{\e}_i = \frac{ \left[\hat \Eb_i \hat \Eb_i^H\right]^{-1} \e_i}{\e_i^H \left[\hat \Eb_i \hat \Eb_i^H\right]^{-1} \e_i}.
 \label{eqnn6}
\end{equation}
However, note that specific to our application, even for a $200 \times 200$ planar array, since the position vectors are of dimensions $40000 \times 1$, the dimension of the matrix $\hat \Eb_i \hat \Eb_i^H$ is $40000 \times 40000$,
and hence ill-conditioned. To circumvent ill-conditioning, one can apply the singular value decomposition (SVD) of $\hat \Eb_i$ to obtain the inverse in the numerator of (\ref{eqnn6}). Let $\U$ and $\D$ be the left singular value matrix and diagonal matrix of singular values, respectively. Since there are $N_I+1$ users, there will be $N_I+1$ non-zero singular values. If $\D_{N_I+1}$ and $\U_{N_I+1}$ are the diagonal matrix of non-zero singular values and the corresponding left singular value matrix, then we can compute $\tilde{\e}_i$ in (\ref{eqnn6})  as 
	\begin{align}
		\tilde{\e}_i = \frac{\U_{N_I+1}\D_{N_I+1}^{-2}\hat{\e_i}}{\hat{\e_i}^H\D_{N_I+1}^{-2}\hat{\e_i}^H},
		\end{align}
		where, $\hat{\e_i} = \U_{N_T+1}^H\e_i$.
Note that, the MPDRB relies on all the user locations in the minimization criterion, including the actual desired receiver location. Similar to the case of NSB, additional derivative constraints can be added to the MDPRB formulation to make it robust. However, evaluating the solution in this case will be time-consuming using popular software such as MATLAB, owing to the $M \times M$ dimensional $\hat \Eb_i \hat \Eb_i^H$.  
	
\par 
 The system's efficacy under the different beamforming schemes is determined by a pair of popular metrics, namely the ASE and the SE, which are functions of various system parameters, like the Rician factor $K$, the array dimensions $M\times M $, or the micro-cell radius, etc. In the next section, we derive the approximate expressions of the performance metrics, followed by extensive simulation results for characterizing the system.

	\section{Theoretical approximations for ASE/SE} 
	
	It is essential for us to characterize the SINR and then derive theoretical expressions for our metrics, such as the ASE and SE. To begin with, let $\alpha_i$ $i=0,1, ..., N_I$ represent the symbol intended for the $i$th user. Without loss of generality, let $i=0$ denote the user under consideration and $i=1,.., N_I$ denote the interferers in the other micro-cells. The symbol received by the user is 
	 
	\begin{equation}
	y= \sqrt{P_r} \h_{0,Ric}^H \tilde{\e}_0  \alpha_0 +  \sqrt{P_r}\sum_{i=1}^{N_I} \h_{0,Ric}^H \tilde{\e}_i \alpha_i + n, 
	\end{equation}
where $n$ represents the complex Gaussian noise having the power of 
	$\sigma^2 = kTBN_F$, with $k = 1.374\times 10^{-23}$ being Boltzmann's constant, $T$ the temperature in Kelvins, $B$ the bandwidth, and $N_F$ the noise figure of the receiver. 
	The received power is given by:
	\begin{align}
	P_r = \frac{P_tG_tG_r}{\tilde{\nu} \nu},
	\label{rec_pow}
	\end{align} 
	where $P_t$ is the power transmitted from the Airliner/HAP, $G_t$ and $G_r$ are the transmitter and receiver antenna gains, while $\tilde{\nu}$ includes the frequency-dependent atmospheric loss also including the back-off loss of the modulation scheme as well as  other transmitter and receiver losses. Finally, the term  $\nu$ represents the path-loss given by \cite{maccartney2017rural},
	\begin{align}
	\nu &= 20\log\left(\frac{4\pi d ~f_c }{c}\right)~dB 
	= \left(\frac{40\pi d~ f_{c, GHz} }{3}\right)^2, 
	%& = 32.4 + 20\log\left(d ~f_{c, Gig}\right)~dB,
	\label{pathloss}
	\end{align}
	where $c=3\times10^8 ~ m/s$ is the speed of the light, $d$ is the distance from the airliner/HAP to the user in meters and $f_{c, GHz}$ is the carrier-frequency in $GHz$.
    The Rician fading channel between the Airliner/HAP and the user is represented by:
	\begin{align}
	\h_{i,Ric} = \sqrt{\frac{K}{1+K}} {\e}_i+\sqrt{\frac{1}{1+K}}h_i\i1,
	\label{p5eqn1}
	\end{align}
	where $K$ is the Rician factor and $h_i$ is the NLoS component, while $\i1$ is the vector of ones. Furthermore, the NLoS component $\h_i$ is a complex Gaussian random variable (RV) with zero mean and unit variance. The instantaneous SINR is now given by the following theorem.
	\begin{thm}\label{thm:sinr}
	The instantaneous signal to interference plus noise (SINR) $\gamma_{SINR}$, of the desired user under NSB is given by \footnote{  Very similar expressions can be derived for the other two schemes, but omitted here given the page limit.  }:
	\begin{align}\label{SINR}
	\gamma_{SINR} &= \frac{P_r |X_s|^2}{\sum_{i=1}^{N_I}P_r|X_i|^2+ \sigma^2},
	\end{align}
	where $X_s \sim \CN(\mu,\sigma_{s}^2)$, with 
	\begin{align}
	\mu = \sqrt{\frac{K}{1+K}} \left(M^2- \e_{0}^H\Eb_0 \left(\Eb_0^H\Eb_0\right)^{-1}\Eb_0^H\e_{0}\right),
	\label{mu}
	\end{align}
	and 
	\begin{align}
	\sigma_{s}^2 = \frac{1}{1+K} \left\|\i1^H\left(1-\Eb_0 \left(\Eb_0^H\Eb_0\right)^{-1}\Eb_0^H\right)\e_0\right\|^2.
	\label{sigma_s}
	\end{align}
	Furthermore, still referring to (\ref{SINR}), we have $X_{i}  \sim \CN(0,\sigma^2_{i})$, where 
	\begin{align}
	\sigma^2_{i} = \frac{1}{1+K} \left\|\i1^H\left(1-\Eb_i \left(\Eb_i^H\Eb_i\right)^{-1}\Eb_i^H\right)\e_i\right\|^2.\label{sigma_i}
	\end{align}
	\end{thm}
	\begin{proof}
	For the proof, please see Appendix \ref{app:sinr}.
	\end{proof}

	Assuming that the users in a cell are allocated identical bandwidths, the Area Spectral Efficiency (ASE) is defined as the sum of the maximum bit rate/Hz/unit
	area supported by the cell's Airliner/HAP \cite{Alouini1999}: 
	\begin{equation}
	ASE = \frac{C}{\pi (D/2)^2},
	\end{equation}
	where  $C$ is the capacity of the intended user in Bps/Hz. Given $\gamma_{SINR}$, the average channel capacity is formulated as:
	\begin{equation}
	\begin{aligned}
	\bar{C}_{SINR} &=  \mathbb{E}[\text{log}_2(1+\gamma_{SINR})]\\
	&= \int_{0}^{\infty}\text{log}_2(1+\gamma_{SINR}) f(\gamma_{SINR}) \text{d}\gamma_{SINR},
	\end{aligned}
	\end{equation}
	where $E[.]$ represents the expectation and $f(\gamma_{SINR})$ denotes the pdf of $\gamma_{SINR}$. By applying the popular use-and-forget bound of \cite{Marzetta2016}, the approximate average channel capacity is given by
	\begin{equation}
	\begin{aligned}
	\bar{C}_{SINR}^{app} &=  \mathbb{E}[\text{log}_2(1+\gamma_{SINR})] \\
	&\approx 
	\text{log}_2\left(\frac{P_r \mathbb E[|X_s|^2]}{\sum_{i=1}^{N_I}P_r\mathbb E[|X_i|^2]+ \sigma^2}\right),
	\end{aligned}
	\end{equation}
	where we have $E[|X_s|^2]= \sigma_s^2+ \mu^2$ and $E[|X_i|^2]= \sigma_i^2$. 
	Thus, the average ASE in $bps/Hz/m^2$ is formulated as:
	\begin{equation}\label{ase_sinr}
	 {ASE}_{SINR}^{app} = \frac{4 \bar{C}_{SINR}^{app}}{\pi D^2}.
	\end{equation}
	Similarly, the average SE of the user is given by
	\begin{equation}\label{se_sinr}
	 {SE}_{SINR}^{app} = \bar{C}_{SINR}^{app}.
	\end{equation}
	
Note that the average capacity is a function of $\mu$, $\sigma_s^2$ and $\sigma^2_i$, $i=1,..., N_I$, which are parameterized by $\theta_i^z$ and $\theta_i^a$ for $i=0,..., N_I$, the zenith and azimuth angles of all the users. Furthermore, the angles themselves are functions of the relative locations of the desired user and the interferers through (\ref{thetaiz}) and (\ref{thetaia}), respectively. Therefore, the total average capacity is obtained by
averaging the expressions over the user locations. Recall that the beamforming vectors are only dependent on the user positions, but not on the channel-gains.   Note that, the approximations are universal, straightforward and are applicable for any value of $K$, received power $P_r$ and other system parameters. 	%Note that the limits for $x_d$ and $y_d$ are between $(-R,R)$, whereas the limits of each of the interference components are between $(3R \cos (360/N_{I,T}), 5R \sin (360/N_{I,T}))$, where $N_{I,T}$ is the number of interferences in a particular tier. 

	\section{Simulation Results}
	\begin{table}[h]
		\begin{center}
			\begin{tabular}{|l|r|}
				\hline
				Parameter & Value \\
				\hline 
				\hline 
				Macro-cell radius $R$& $5-20 ~km$ \\
				\hline
				Micro-cell radius $r$& $50~m$, $75~m$, $100~m$ \\
				\hline
				Vertical airliner/HAP distance $H_t$& $10~km$, $21~km$ \\
				\hline
				Carrier frequency $f_c$ & $73.5~GHz$ \\
				\hline 
				Total Bandwidth B & $5~GHz$ \\
				\hline 
				Reuse factor & 7 \\
				\hline 
			    Bandwidth per User & $714~MHz$ \\
				\hline 
				Dimensions of the planar array $M$ & $200$, $300$, $400$, $500$\\
				\hline
				Rician factor $K$ & $10,15,30~dB$\\
				\hline
				Back-Off & $10~dBm$ \\
				\hline
				Transmitter loss & $1.8~dB$ \\
				\hline
				Transmitter antenna gain & $10\log (M^2)$ \\
				\hline
				Atmospheric and cloud loss & $7.9~dB$ \\
				\hline 
				Receiver antenna gain & $60.2~dB$ \\
				\hline 
				Receiver noise figure & $6~dB$ \\
				\hline 
				Other receiver loss & $1.8~dB$ \\
				\hline
			\end{tabular}
		\end{center}
		\caption{Simulation parameters}
		\label{table:TxParams}
	\end{table}

	\begin{figure}[h]
\centering
  \includegraphics[width=0.5\textwidth]{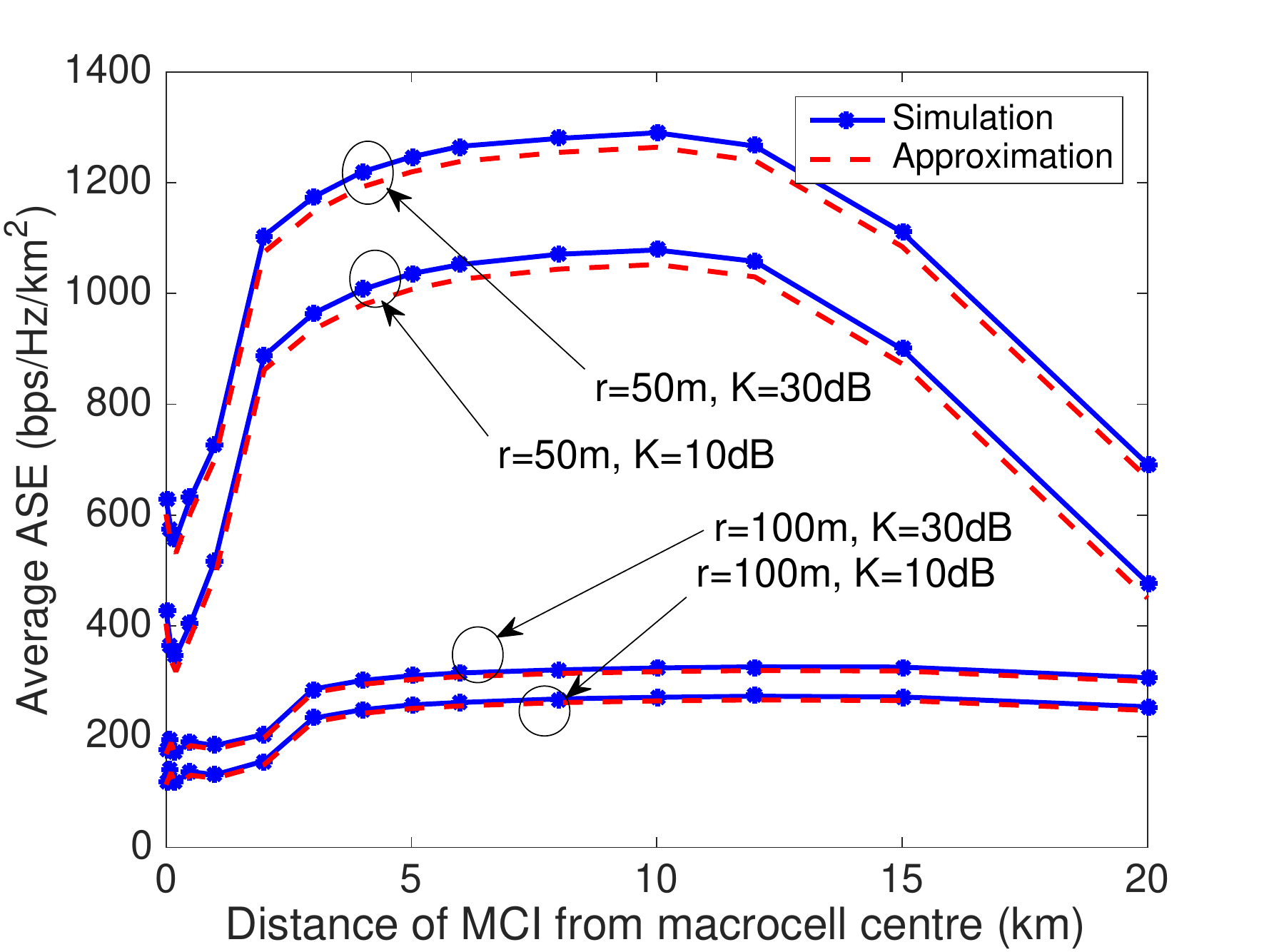}
\caption{Average ASE of NSB vs. distance of MCI from the macro-cell centre for $M=500$ and $H_t=10~km$. The theoretical result is based on (\ref{ase_sinr}).}
\label{fig:ase1}
\end{figure}

\begin{figure}[h]
\centering
  \includegraphics[width=0.5\textwidth]{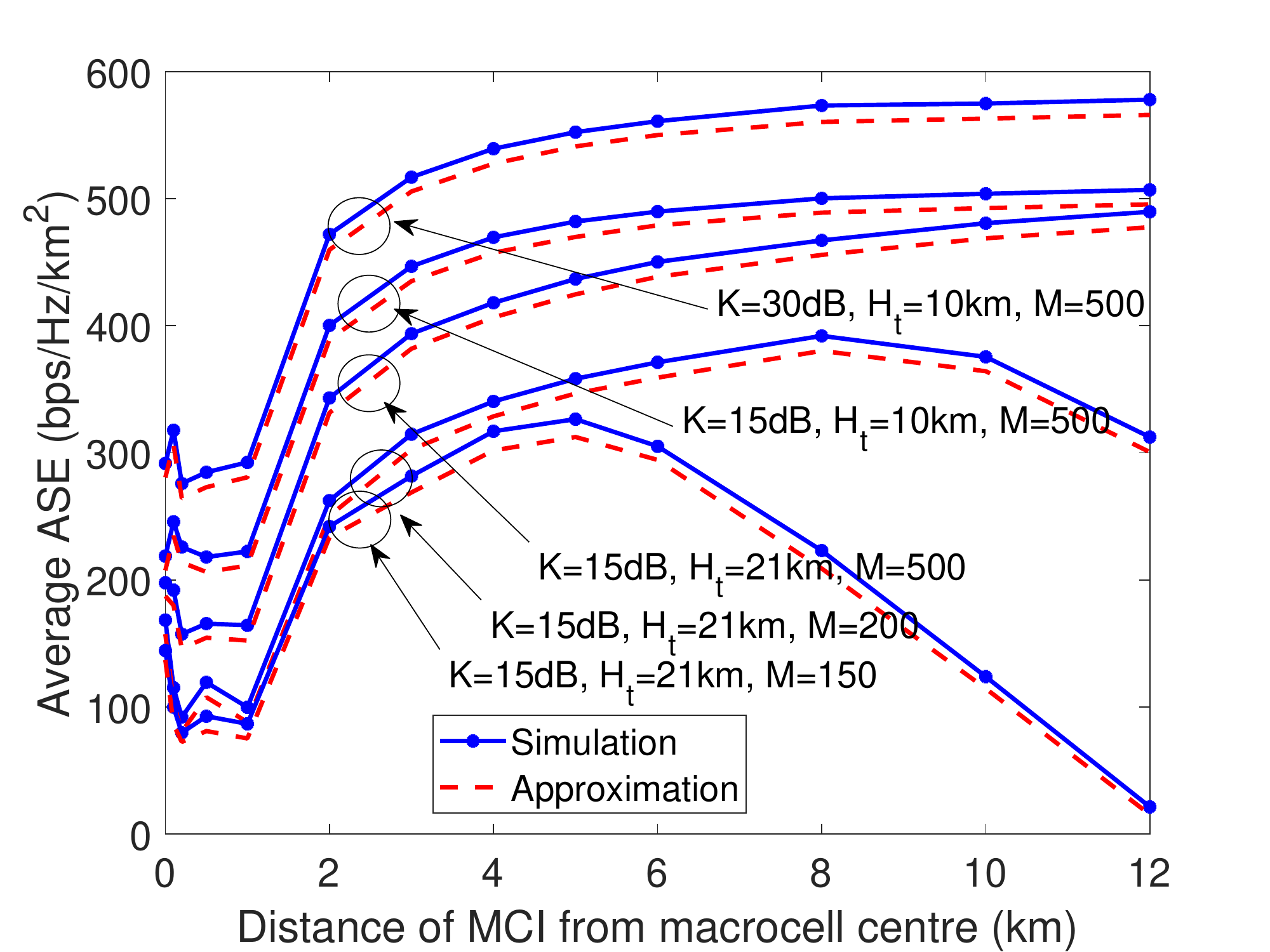}
\caption{Average ASE of NSB vs. distance of MCI from the macro-cell centre for $r=75~m$. The theoretical result is based on (\ref{ase_sinr}).}
\label{fig:ase2}
\end{figure}

Extensive Monte-Carlo simulations have been carried out for evaluating the ASE and SE of the system, assuming that the Airliner/HAP is located at $(0,0,0)$.  The desired UE is positioned uniformly in the MCI of radius $r$ with its centre located at $(x,y,-H_t)$. For a frequency reuse factor of $7$, the reuse distance is fixed at approximately $D=4r$. We consider the five tiers of interfering micro-cells with their centres set to $4kr, \forall k=1,..,5$ from the MCI\footnote{The number of interfering tiers to be chosen is based on an engineering trade-off that varies with the distance from the macro-cell centre. For example, for MCI at the centre of macro-cell (directly below the airliner), considering three tiers of interferers is sufficient. In the case of the macro-cell edge, 5 tiers of interferers are needed for $r=50~m$.}. The interfering UEs are also uniformly placed in the interfering micro-cell of radius $r$. The interferences imposed by the more distant tiers of micro-cells and macro-cells are assumed to be negligible. To represent a strong LoS component, we consider the Rician factors $K$ to be $10$, $15$ and $30$ dB. All the other transmission parameters were proposed initially in \cite{Huang2019} and are summarized in Table \ref{table:TxParams} for completeness.
	
\par In Fig. \ref{fig:ase1}, the simulated and approximate ASE of NSB are plotted vs. the horizontal distance of the MCI from the airliners, for several Rician factors $K$ and micro-cell radii $r$. Naturally, upon increasing $K$, the ASE/SE increases. Since the ASE is inversely proportional to the cell-radius, it decreases upon increasing the micro-cell radius $r$. However, the ASE fails to reach its maximum, when the MCI is directly below the airliner, namely when the MCI is at the macro-cell centre. The ASE is a function of both the MCI distance from the airliners as well as of $ K $ and of the micro-cell radius $r$. For example, for $r=50~m$ and $K=30$ dB, an ASE as high as $1200~ bps/Hz/km^2$ is achieved when the user is at a distance of $5~km$ from the macro-cell centre.\footnote{The signal and interference powers are approximately of the order of $10^6$ and $10^{-6}$ respectively. At a reuse distance $D=200~m$, the SE and ASE are $38~ bps/Hz$ and $1200~ bps/Hz/km^2$ approximately.} At the macro-cell centre, the ASE is reduced to $600~ bps/Hz/km^2$.

When the desired user is at the macro-cell centre, the interferers are distributed in all the quadrants; hence their azimuth angles are uniformly distributed in $[0, 360]$ degrees. Now, as the desired user moves away from the macro-cell centre, there are two effects. Firstly, the range of azimuth angles of the interferers decreases because the first five tiers of interferers we consider are concentrated in a single quadrant. Secondly, the absolute value of the zenith angles ($\theta^z$) of the interference increases from $0^0$ as we move away from the macro-cell centre. Both these effects increase the correlation between the desired and interfering users' steering vectors, hence reducing the power of the null-steering vectors. Therefore, both the signal and interference powers are reduced as the micro-cell centre moves away from the macro-cell centre. Near the macro-cell centre the reductions of both the signal and interference powers become similar and hence the ASE fluctuates. However, as the micro-cell centre moves further away, the interference power reduction is more substantial than the signal power reduction and hence the ASE increases. Further away, the desired power reduction becomes substantial and hence the ASE decreases. The SE variations vs. the distance can also be explained using similar reasoning. Furthermore, for higher altitudes of the airliner, the SINR of the users decreases owing to their higher path-loss. Therefore, the ASE decreases, as observed in Fig. \ref{fig:ase2}. Note that the ASE achieved by the massive MIMO scheme of \cite{Xin2015} is on the order of $10~bps/Hz/km^2$, while our scheme achieves in the order of $1000~ bps/Hz/km^2$, which is comparable to the ASE achieved by the HetNet and DenseNet of \cite{Li2016} and \cite{Ding2015}, respectively.   The transmit (Tx) power at each of the antenna elements needs to be carefully chosen depending on the noise power, the number of antenna elements, and the power consumption allowed at the transmitter. For a Tx power of $30~dBm$, power consumption of $40~kW$ is incurred for $M=200$. Whereas for a Tx power of $5~dBm$, power consumption of only $0.126~kW$ is incurred, without any drop in the ASE \footnote{Though a smaller Tx power can be chosen with no drop in the ASE, we have chosen $5~dBm$ to account for unforeseen practical losses.}. Even if we choose $M=500$, the power consumption does not exceed $0.8~kW$ for a Tx power of $5~dBm$.

\begin{figure}[h]
\centering
  \includegraphics[width=0.5\textwidth]{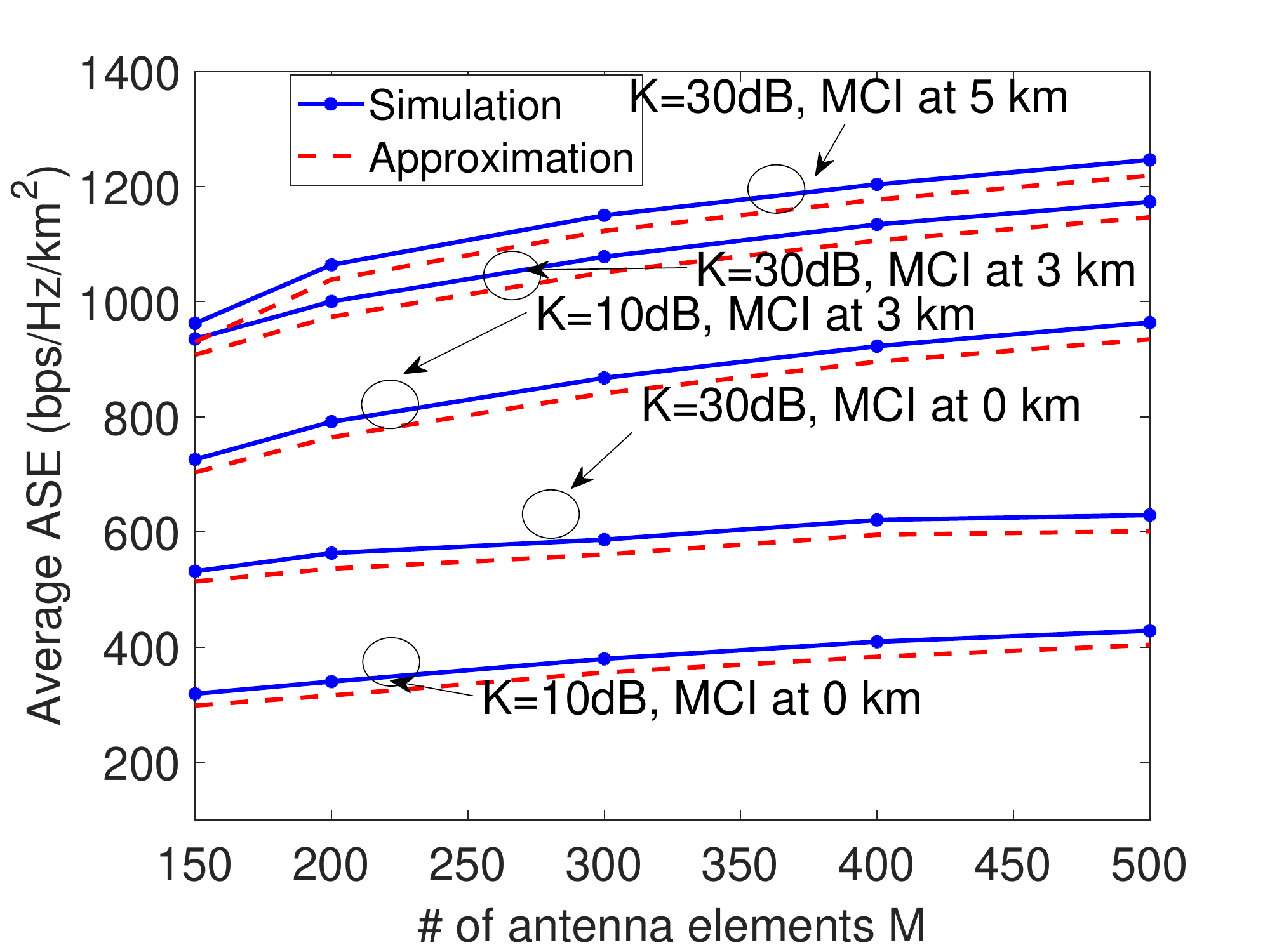}
\caption{Average ASE of NSB vs. $M$ for $H_t=10~km$ and $r=50~m$. The theoretical result is based on (\ref{ase_sinr}).}
\label{fig:ase3}
\end{figure}

\begin{figure}[h]
\centering
  \includegraphics[width=0.5\textwidth]{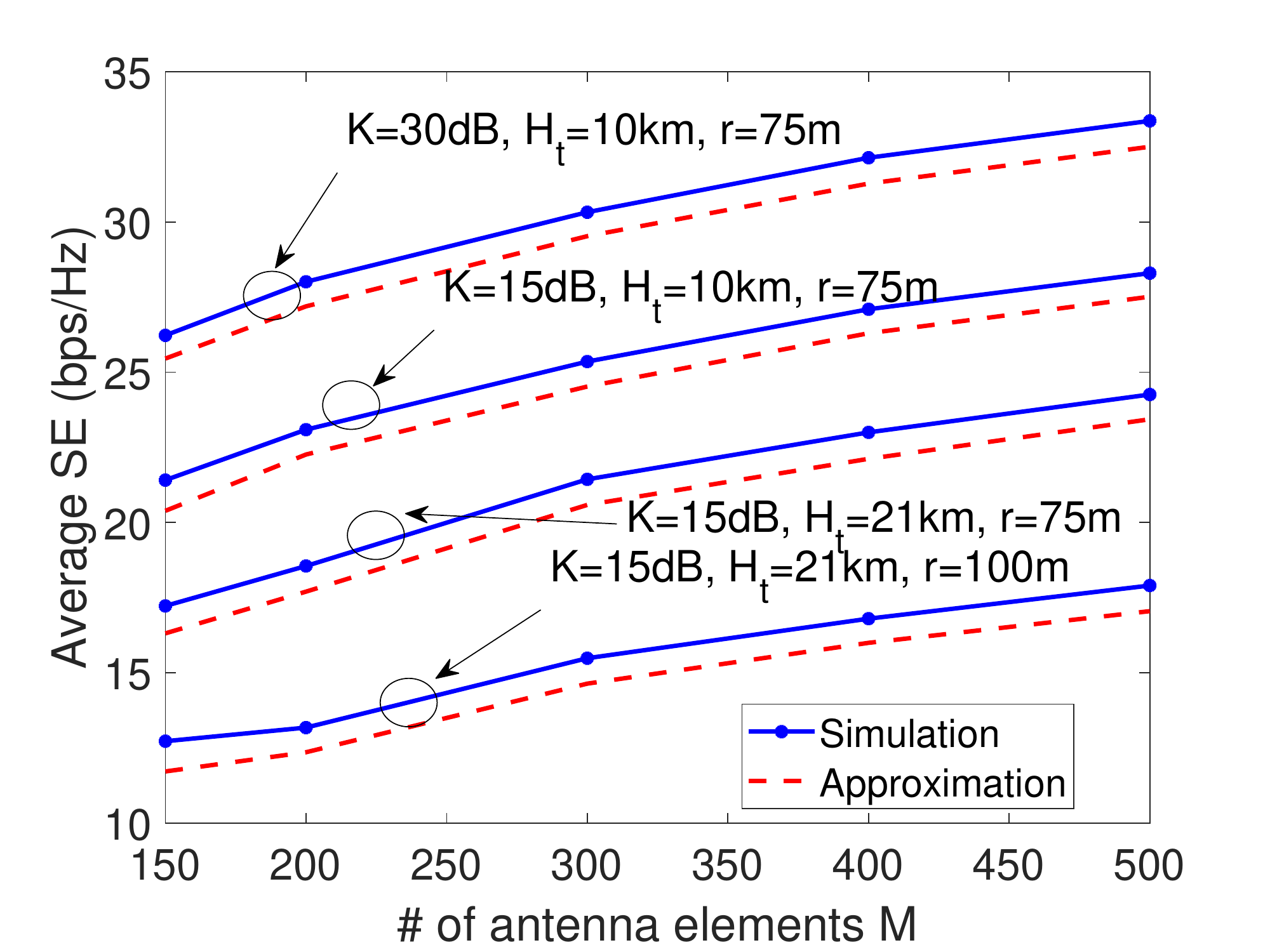}
\caption{Average SE of NSB vs. $M$ for MCI at $2~km$ from the macro-cell center. The theoretical result is based on (\ref{se_sinr}).}
\label{fig:se1}
\end{figure}

\par The ASE and SE variations of NSB vs. $M$ are portrayed in Fig. \ref{fig:ase3} and \ref{fig:se1}. An array dimension of $500\times 500$ provides the highest ASE. For an antenna element spacing of $\lambda/2$, where $\lambda=4~mm$ is the career wavelength, the array dimensions will not exceed $1~m^2$. The increase in ASE vs. $M$ remains marginal compared to that vs. the Rician factor $K$. For example, for an increase in $M$ from $200$ to $500$, the ASE improves from approximately $300$ to $400$ $bps/Hz/km^2$ respectively. On the other hand, for an increase in $K$ from $10$dB to $30$dB, the ASE improves from $300$ to $550$ $bps/Hz/km^2$, respectively. With an increase in micro-cell radius, we observe a SE reduction in Fig. \ref{fig:se1}. However, the reduction is only marginal compared to the ASE reduction seen in Fig.\ref{fig:ase1}.

\begin{figure}[h]
\centering
  \includegraphics[width=0.5\textwidth]{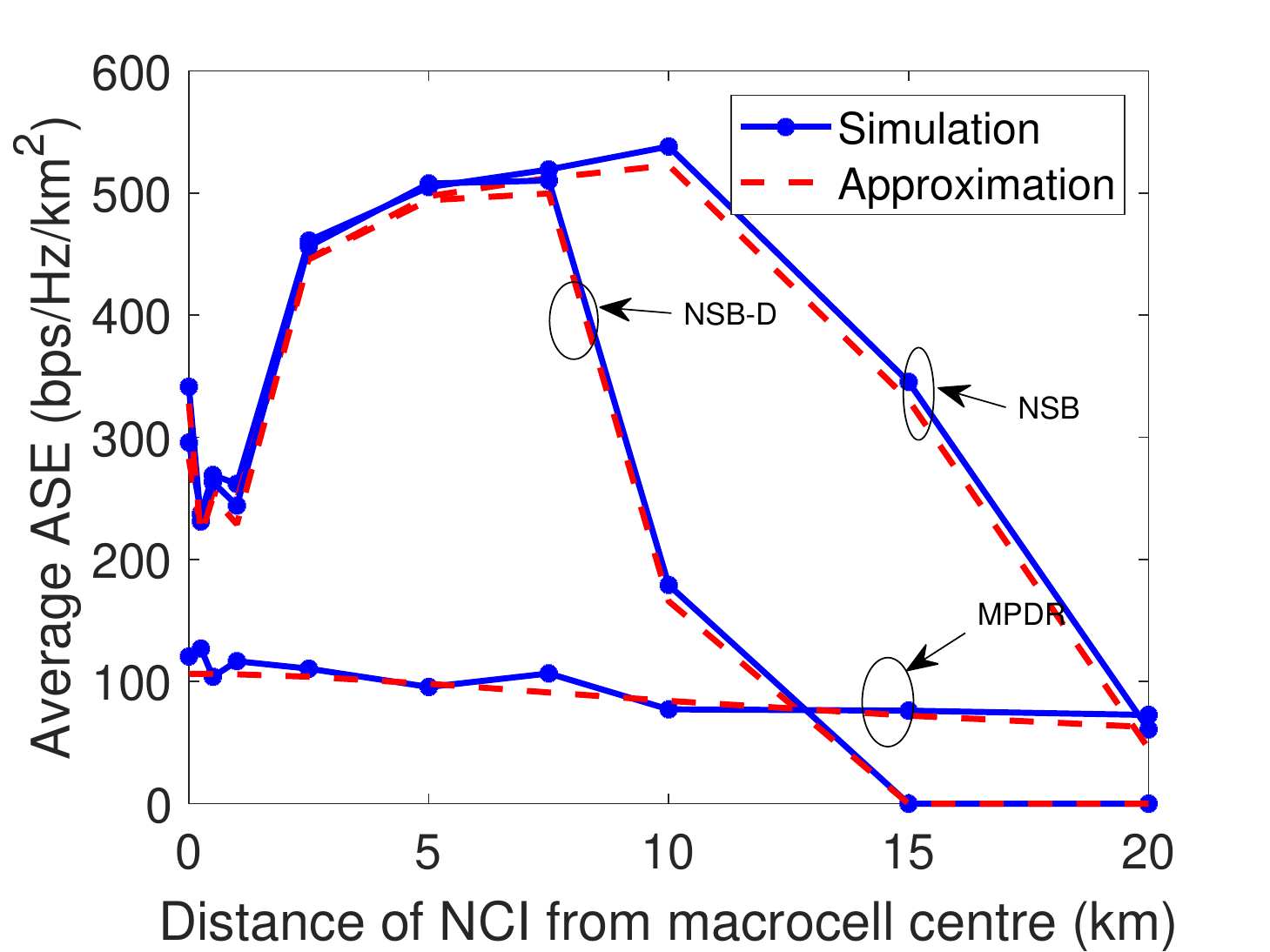}
\caption{Average ASE vs. distance of MCI from the macro-cell centre for $M=300$, $r=75~m$, $K=30dB$ and $H_t=10~km$.}
\label{fig:ase4}
\end{figure}

\begin{figure}[h]
\centering
  \includegraphics[width=0.5\textwidth]{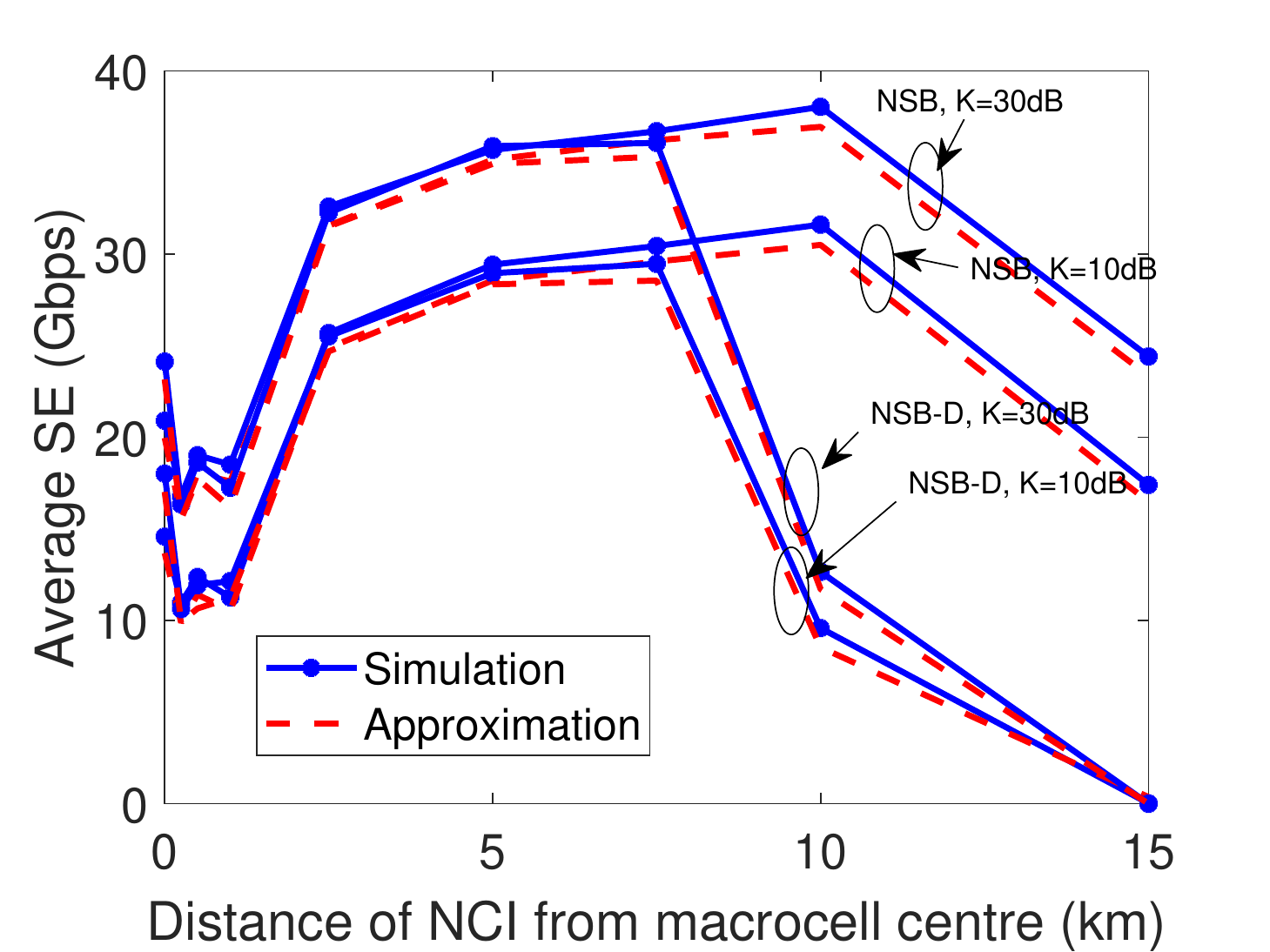}
\caption{Average SE vs. distance of MCI from the macro-cell centre for $M=300$, $r=75~m$ and $H_t=10~km$.}
\label{fig:se2}
\end{figure}

In Fig. \ref{fig:ase4}, the simulated ASE is plotted vs the horizontal distance of the MCI from the airliners for various beamforming techniques. We can observe that NSB-D performs nearly as well as NSB, provided that the MCI is near the macrocell centre. However, recall that as the distance of the MCI increases, the correlation between the desired and interfering users' steering vectors increases, hence reducing the beamforming vectors' power. Coupled with the widening of beam by the derivative constraints, a drastic output SINR, and ASE reduction, is observed. For example, for an MCI at a distance of $10$~km from the macrocell centre, the ASE of NSB-D is nearly half of that of NSB.  A similar trend can be observed in SE for different $K$ values, as seen in Fig. \ref{fig:se2}.
On the other hand, MPDRB provides the lowest ASE of the three beamformers. For example, for an MCI at a distance of $2.5$~km from the macrocell centre, the ASE of MPDRB is nearly third of that of NSB. However, note that MPDRB is more robust to variations in the distance of the MCI from the macrocell centre. The correlation between the desired and interfering users' steering vectors does not contribute to a significant ASE reduction because it does not explicitly create nulls in the interfering users' directions, unlike NSB and NSB-D.

\subsection{Doppler compensation}
  We note that all the above simulations considered an idealized setting for which the Doppler spread was neglected. However, in a practical system, it is necessary to compensate for the Doppler frequency offset (DFO) caused by the airliner's radial velocity component with respect to the user location. Since the velocity reading of the high-speed airliner are perfectly known at the transmitter, the Doppler shift due to this motion can be pre-compensated~\cite{guo2017angle}. Assuming that the airliner is heading in the $0^0$ azimuth-direction, the radial component of the velocity vector towards the angle pair $(\theta_i^z,\theta_i^a)$ is 
	\begin{align}
	v_r = v_a\cos \theta_i^a \sin \theta_i^z,
	\label{p6eqnn1}
	\end{align}
	where $v_a$ is the airplane's speed. Hence, the Doppler-induced frequency deviation in the $(\theta_i^z,\theta_i^a)$-direction is given by
	\begin{align}
	f_D = \frac{\left(1 - \frac{v_r}{c}\right)}{\left(1 + \frac{v_r}{c}\right)}f_T,
	\label{p6eqnn2}
	\end{align}
	where $c$ is the speed of light and $f_T$ is the carrier frequency. To pre-compensate for the spread in frequency, the airplane can transmit at a frequency of
	\begin{align}
	{f}_T = \frac{\left(1 + \frac{v_r}{c}\right)}{\left(1 - \frac{v_r}{c}\right)}f_c,
	\label{p6eqnn3}
	\end{align}
	so that at the user end, the received signal frequency will be $f_c$ after Doppler shift. Note that for an average airplane speed of $200~$m/s and a high carrier frequency of $73.5GHz$, the Doppler offset is negligible, since $\left(1 \pm \frac{v_r}{c}\right) \approx 1$. 
	In some cases, the radial component of the velocity vector towards the angle pair can be incorrectly estimated as 
	\begin{align}
\tilde {v_r} = v_a\cos \theta_i^a \sin \theta_i^z + \Delta v_r v_a\cos \theta_i^a \sin \theta_i^z,
	\label{p6eqnn4}
	\end{align}
	where $\Delta v_r$ is the offset/error in the estimation with $\Delta v_r=0$ indicating perfect estimation. In such a case, we can observe that the pre-compensation results in the airplane transmitting at a frequency of 
	\begin{align}
	\hat{f}_T = \frac{\left(1 + \frac{\tilde v_r}{c}\right)}{\left(1 - \frac{\tilde v_r}{c}\right)}f_c,
	\label{p6eqnn5}
	\end{align}
	which after the Doppler spread results in 	\begin{align}
	\hat{f}_c = \frac{\left(1 - \frac{v_r}{c}\right)}{\left(1 + \frac{v_r}{c}\right)}\frac{\left(1 + \frac{\tilde v_r}{c}\right)}{\left(1 - \frac{\tilde v_r}{c}\right)}f_c
	\label{p6eqnn6}
	\end{align}
	at the user end.  In Table \ref{tab:DFO}, the average ASE is tabulated for the different beamformers operating with and without incorrect frequency offset estimation. It can be noted that even for a $50\%$ error in estimating the offset, corresponding to $\Delta v_r= \pm 1$, the degradation in the ASE is moderate. The degradation is the lowest for NSB-D, because the derivative constraints increase the beam-width towards the user directions and increase the robustness to steering vector mismatches.

\begin{table}[h]
    \centering
    \begin{tabular}{|c|c|c|c|c|c|}
    \hline
    Beamformer & $\Delta v_r=-1$ & $-0.5$ & $0$ & $0.5$ & $1$\\
    \hline
     \hline
     NSB & 963 &  969  & 969 & 969 & 965\\
     \hline
     NSB-D  & 907 & 908 & 909 & 908 & 907\\
         \hline
     MPDRB  & 215 & 215 & 217 & 216 & 215\\
        \hline
        \hline
    \end{tabular}
    \caption{Average ASE expressed in $bps/Hz/km^2$ vs $
    \Delta v_r$ for $M=200$, $K=30~dB$, $r=50$~m, $H_t=10$~km and MCI at a distance of $2.5$~km from the macrocell centre.}
    \label{tab:DFO}
\end{table}

\subsection{Position vector mismatches}
  One of the essential requirements for the proposed technique is the signalling of the user locations at the airplane. The users can estimate their location using a global positioning system (GPS) and communicate it to the airplane with a local base station's aid. A high precision technique such as real-time kinematic (RTK) processing or carrier phase differential tracking can be employed \cite{dai2001study}. Alternatively, it is possible to continuously track the user from the airplane by periodically sending signals to the users. In other words, the established methods of precise positioning of high-speed vehicles can be used to get an accurate estimate of the user location  \cite{talvitie2020beamformed,talvitie2018positioning}. Given precise estimates of the users' locations and the airplane's velocity, sub-meter position accuracy can be obtained for accurate beamforming \cite{lenschow1972measurement}.

We study the performance of the beamformers for up to $5$~m offset in the users' position. Let us assume that the coordinates of the users and the interferers are incorrectly measured as $(x_i+\delta cos \beta_i,y_i +\delta sin \beta_i,-H_t)$ $i=0,1,..., N_I$, where $\delta$ is the magnitude of the offset and $\beta_i$ is the random phase. The average ASE for various values of $\delta$ (in metres)  is shown in Table \ref{tab:PVM}. We observe that NSB-D is more robust to position vector mismatches than NSB. NSB-D also outperforms NSB in the presence of position vector mismatches due to its broader beamwidth. For example, for an MCI at $3.5$~km, with no position vector mismatch, i.e., $\delta=0$, the ASE of NSB-D is $763~bps/Hz/km^2$, and that of NSB is $1025~bps/Hz/km^2$. For $\delta=1$~m, the ASE is nearly halved in the case of NSB, whereas the degradation is negligible for NSB-D. Although MPDRB is also resistant to position vector mismatches compared to NSB, the overall ASE obtained still remains lower.

\begin{table}[h]
    \centering
    \begin{tabular}{|c|c|c|c|c|c|}
    \hline
    Beamformer & MCI (km) & $\delta= 0$  & $\delta= 0.5$ & $\delta= 1$ & $\delta= 5$\\
    \hline
     \hline
     NSB & $1$ &  637  & 528 & 474 & 329\\
      & $2.5$ &  969 & 544  & 480 & 335\\
      & $3.5$ &  1025 & 538 & 478 & 332\\
     \hline
     NSB-D  & $1$ & 580 & 580 & 580 & 550\\
       & $2.5$ & 908 & 908 & 867 & 599\\
     & $3.5$ & 763 & 763 & 746 & 496\\
         \hline
        MPDRB  & $1$ & 221 & 221 & 221 & 220\\
       & $2.5$ & 220 & 220 & 218 & 217\\
       & $3.5$ & 215 & 215 & 214 & 214\\
        \hline
        \hline
    \end{tabular}
    \caption{Average ASE expressed in $bps/Hz/km^2$ vs $\delta$ for $M=200$, $K=30~dB$, $H_t=10$~km and $r=50$~m.}
    \label{tab:PVM}
\end{table}

\par Both the theory and simulations indicate substantial ASE gains for our proposed airplane-aided integrated network. However, our work only represents the first step towards realizing a high-capacity ISTN; hence it relies on some idealized simplifying assumptions to be eliminated by future research.

\section{Conclusions and Directions for Future Research} 
\par In this treatise, we considered a planar-array aided stand-alone airliner/HAP in a macro-cell communicating with the terrestrial BS/users.   To achieve a high directional gain and to minimize the interference among the users, we invoked three beamforming techniques, namely NSB, NSB with derivative constraints, and MPDRB, for transmission from the airliner/HAP and provided approximate SE and ASE expressions. We also studied the performance of the system for incorrect Doppler compensation and position vector mismatches. NSB was the least complex TPC scheme that provided maximum ASE, in the absence of position vector mismatches. However, NSB-D was more robust to position vector mismatches and MPDRB more robust to variations in the distance of the MCI from the macrocell centre.
 
\par  We considered a simple system, where the airliner is the network provider for the terrestrial users. The routing protocols, traffic-offloading and optimizing the tele-traffic resources in conjunction with the existing terrestrial and space networks require careful further study.    An interesting future study would quantify the effect of the beam-squint encountered by the proposed schemes \cite{liao2021terahertz, wang2018spatial, chen2020hybrid}. A possible solution to overcome this problem could be to partition the entire band into multiple narrower bands and apply a phase shift based beamformer separately. Furthermore, the hybrid beamforming techniques discussed in \cite{chen2020hybrid, ahmed2018survey, bogale2016number, han2015large, chen2017efficient, wang2018digital, 8334262} could be adopted for overcoming the beam-squint effects and for reducing the power consumption.   Another critical topic that can be explored is the impact of mobility and real flight schedule on the handovers between the different macro-cells and the existing networks. Finally, the optimal radius of micro-cells used at different distances from the macro-cell centre is another exciting future research direction.

	\appendices
	\section{Proof for Theorem \ref{thm:sinr}}\label{app:sinr}
The array gain of the desired user is $X_s = \h_{0,Ric}^H \tilde{\e}_0$ and that of the $i$th interferer is $X_i=\h_{0,Ric}^H \tilde{\e}_i $. In order to determine the distribution of the SINR, we have to determine the distribution of the components in the numerator and the denominator. Upon expanding $X_s$, we arrive at, 
	\begin{align}
	X_s &=  \h_{0,Ric}^H \tilde{\e}_0 \nonumber \\
	&= \sqrt{\frac{K}{1+K}} \left(\e_0^H \e_0- \e_{0}^H\Eb_0 \left(\Eb_0^H\Eb_0\right)^{-1}\Eb_0^H\e_{0}\right) \nonumber \\
	& \: + \sqrt{\frac{1}{1+K}}h_0\i1^H\left(1-\Eb_0 \left(\Eb_0^H\Eb_0\right)^{-1}\Eb_0^H\right)\e_0.
	\end{align}
	Observe that by construction, we have: 
	\begin{equation}
	\e_j^H \e_i= \begin{cases}
	M^2;  \qquad \qquad \quad  i=j\\
	\frac{\sin \left( \frac{M}{2} \left(\psi_{j}^x -\psi_i^x\right)\right)}{ \left( \frac{1}{2} \sin \left(\psi_{j}^x -\psi_i^x\right)\right)} \frac{\sin \left( \frac{M}{2} \left(\psi_{j}^y -\psi_i^y\right)\right)}{ \left( \frac{1}{2}\sin \left(\psi_{j}^y -\psi_i^y\right)\right)};\: i \neq j,
	\end{cases}
	\end{equation}
	where $\psi_x$ and $\psi_y$ are defined in (\ref{psix}) and (\ref{psiy}), respectively. Since, $h$ is a zero-mean Gaussian RV with unit variance, it can be readily seen that $X_s$ is a Gaussian RV with a mean given by (\ref{mu}) and variance by (\ref{sigma_s}).
	Now the $i$th interference component is formulated as:
	\begin{align}
	X_i &= \h_{0,Ric}^H \tilde {\e}_i,  \nonumber \\
	&= \left(\sqrt{\frac{K}{1+K}}\e_{0}^H+\sqrt{\frac{1}{1+K}}h_0\i1^H \right)\tilde{\e}_i \nonumber\\
	&= \left(\sqrt{\frac{K}{1+K}}\e_{0}^H+\sqrt{\frac{1}{1+K}}h_0\i1^H \right) \nonumber \\
	& \quad \times \left(1-\Eb_i \left(\Eb_i^H\Eb_i\right)^{-1}\Eb_i^H\right)\e_i,
	\end{align}
where the null-steering beamforming vectors $\tilde{e}_j$ are designed in such a way that $\e_{i}^H \tilde{\e}_j \approx 0$, $\forall~i~\ne~j$. Hence, the RV $X_i$ has a negligible mean component. On the other hand, $X_i$ has a variance given by (\ref{sigma_i}).

\end{document}